\documentclass[a4paper,onecolumn,11pt,accepted=2022-07-28]{quantumarticle}
\pdfoutput=1
\usepackage[utf8]{inputenc}
\usepackage[english]{babel}
\usepackage[T1]{fontenc}
\usepackage{amsmath}
\usepackage{hyperref}

\usepackage{tikz}
\usepackage{lipsum}

\usepackage{fullpage}
\usepackage{amsfonts,amsthm,mathrsfs, mathtools,xspace,graphicx}
\usepackage{endnotes}
\usepackage{color}
\usepackage{bbm}
\usepackage{amssymb,latexsym}
\usepackage{caption,subcaption}
\usetikzlibrary{decorations.markings}

\newtheorem*{theorem*}{Theorem}
\newtheorem{theorem}{Theorem}[section]
\newtheorem{proposition}[theorem]{Proposition}

\newtheorem{lemma}[theorem]{Lemma}
\newtheorem{claim}[theorem]{Claim}

\theoremstyle{remark}
\newtheorem{remark}[theorem]{Remark}

\theoremstyle{definition}

\newcommand{\beq}{\begin{eqnarray}}
\newcommand{\eeq}{\end{eqnarray}}

\newcommand{\ket}[1]{|#1\rangle}
\newcommand{\bra}[1]{\langle#1|}
\newcommand{\proj}[1]{\ket{#1}\!\bra{#1}}
\newcommand{\Tr}{\mbox{\rm Tr}}
\newcommand{\Id}{\ensuremath{\mathop{\rm Id}\nolimits}}
\newcommand{\Es}[1]{\ensuremath{\mathop{\textsc{E}}_{#1}}}

\newcommand{\reg}[1]{{\textsf{#1}}}
\newcommand{\ol}[1]{\overline{#1}}

\newcommand{\C}{\ensuremath{\mathbb{C}}}

\newcommand{\F}{\ensuremath{\mathbb{F}}}

\newcommand{\mA}{\ensuremath{\mathcal{A}}}
\newcommand{\mB}{\ensuremath{\mathcal{B}}}

\newcommand{\mH}{\mathcal{H}}

\newcommand{\supp}{\textsc{Supp}}

\newcommand{\norm}[1]{\left\| {#1} \right\|}

\def\*#1{\mathbf{#1}}

\DeclarePairedDelimiter\parens{\lparen}{\rparen}

\newcommand{\kb}[2]{\ket{#1}\!\bra{#2}}

\begin{document}

\title{A monogamy-of-entanglement game for subspace coset states}

\author{Eric Culf}
\affiliation{Department of Mathematics and Statistics, University of Ottawa, Canada}
\email{eculf019@uottawa.ca}

\author{Thomas Vidick}
\affiliation{Department of Computing and Mathematical Sciences, California Institute of Technology, USA}
\email{vidick@caltech.edu}
\maketitle

\begin{abstract}
    We establish a strong monogamy-of-entanglement property for subspace coset states, which are uniform superpositions of vectors in a linear subspace of $\F_2^n$ to which has been applied a quantum one-time pad. This property was conjectured recently by [Coladangelo, Liu, Liu, and Zhandry, Crypto'21] and shown to have applications to unclonable decryption and copy-protection of pseudorandom functions. We present two proofs, one which directly follows the method of the original paper and the other which uses an observation from [Vidick and Zhang, Eurocrypt'20] to reduce the analysis to a simpler monogamy game based on BB'84 states. Both proofs ultimately rely on the same proof technique, introduced in [Tomamichel, Fehr, Kaniewski and Wehner, New Journal of Physics '13]. 
\end{abstract}

\section{Introduction}

Informally, a \emph{monogamy game} is a game in which the maximum success probability is tied to the monogamy of entanglement, i.e.\ limitations on the strength of quantum multipartite correlations. The simplest such game goes as follows. Two players Bob and Charlie aim to prepare a tripartite state $\rho_{\reg{ABC}}$, such that $\reg{A}$ is a single qubit and $\reg{B}$ and $\reg{C}$ are arbitrary, and the following holds: given a measurement of $\reg{A}$ in the standard or Hadamard basis yielding an outcome $x\in\{0,1\}$ it is possible to predict $x$ both by making a measurement on $\reg{B}$ only \emph{and} on $\reg{C}$ only, given the chosen basis as side information. Monogamy of entanglement expresses itself by the fact that while ignoring $\reg{C}$ it is possible to win in this game with probability $1$ by choosing $\rho_{\reg{AB}}$ to be an EPR pair, as soon as $\reg{C}$ is present the maximum winning probability drops to $\frac{1}{2} + \frac{1}{2\sqrt{2}} \approx 0.854$. 

Monogamy games have played an important role in quantum cryptography since some of the first  proofs of security of quantum key distribution, which make use of monogamy through uncertainty relations such as $H(Z|\reg{B})+H(X|\reg{C})\geq 1$, with $X$ and $Z$ classical random variables that denote the outcome of a measurement of $\reg{A}$ in the standard and Hadamard bases respectively~\cite{koashi2006unconditional,tomamichel2017largely}. In this note we study a  monogamy game introduced recently in~\cite{coladangelo2021hidden} and called ``strong monogamy game'' therein. Informally, in the game two players Bob and Charlie cooperate in an attempt to create two copies of a \emph{coset subspace state} 
\[ \ket{A_{s,s'}} \,=\, \frac{1}{\sqrt{|A|}}\sum_{u\in A}(-1)^{u\cdot s'} \ket{u+s}\;,\]
where $A$ is a linear subspace of $\F_2^n$ and $s,s'\in\F_2^n$ are arbitrary, such that given the first copy \emph{and} a description of $A$ it is possible to obtain a vector $u\in A+s=\{a+s\vert a\in A\}$, while given the other copy \emph{and} the description of $A$ it is possible to obtain a vector $v\in A^\perp + s'$, with $A^\perp = \{w:\, w\cdot u=0\,\forall u\in A\}$.\footnote{Here it is crucial that $A$ is revealed only after the ``copying'' has taken place, as given $A$ and $\ket{A_{s,s'}}$ itself it is possible to recover $s\mod A$ and $s'\mod A^\perp$.} (We describe the game in detail in Section~\ref{sec:coset-game}.)
In~\cite{coladangelo2021hidden} the authors show a sub-exponentially decaying bound on the players' maximum success probability in a variant of this game where from each copy a pair $(u,v)\in (A+s)\times (A^\perp + s')$ has to be returned. While the original subspace coset game is more useful for their cryptographic applications they are unable to analyze it. In this paper we show an exponentially decaying bound on the players' maximum success probability in the original game; as shown in~\cite{coladangelo2021hidden} this implies constructions for uncloneable decryption and copy-protection of pseudorandom functions based on post-quantum indistinguishability obfuscation and one-way functions only. (In contrast, in~\cite{coladangelo2021hidden} the same applications are obtained under the additional, strong assumption of extractable witness encryption. We refer to~\cite{coladangelo2021hidden} for additional discussion.)    

Our main result is stated as Theorem~\ref{thm:coset} in Section~\ref{sec:coset-game}. We first show the theorem directly by following the template introduced in~\cite{tomamichel2013monogamy} and adapting it to subspace coset states using some of the arguments from~\cite{coladangelo2021hidden} as well as some new steps. It is interesting to note that the direct proof does not make much use of the particular structure of the subspaces, rather just the fact that the states are constructed from cosets. As such, it might be possible to generalise this monogamy property to cosets states for a much larger class of groups, such as those in~\cite{ACP20}. Next, we revisit our direct proof by making a simple but useful connection between subspace coset states and BB'84 states. (This connection was first used in~\cite{vidick2021classical} to analyze a proof of quantum knowledge for subspace coset states.) To explain the connection, let $A$ be a subspace spanned by canonical vectors, $A = \textrm{Span}\{e_i,\, i\in \ol{T}\}$ for some set $T\subseteq \{1,\ldots,n\}$ with complement $\ol{T}$, and $s,s'\in \F_2^n$. Let $\theta\in\{0,1\}^n$ be the indicator vector of $\ol{T}$, i.e.\ $\theta_i=1$ if and only if $i\notin T$. Let $x\in\{0,1\}^n$ be such that $x_i=s_i$ whenever $i\in T$ and $x_i=s'_i$ whenever $i\in \ol{T}$. Then it is easily verified that 
\[ \ket{A_{s,s'}} \,=\, \ket{x}_\theta\;,\]
where we write $\ket{x}_\theta = \ket{x_1}_{\theta_1}\cdots\ket{x_n}_{\theta_n}$ with $\ket{x_i}_{\theta_i} = H^{\theta_i}\ket{x_i}$, $H$  the Hadamard gate. Thus coset subspace states for ``basis-aligned'' subspaces are exactly BB'84 states. This observation leads to a partition of subspace coset states such that subspace coset states in each element of the partition are in $1$-to-$1$ correspondence with BB'84 states under a simple unitary permutation of the standard basis, see Claim~\ref{claim:translation} for a precise formulation. While this observation implicitly appears in some of the arguments from~\cite{coladangelo2021hidden}, as well as in our direct proof of Theorem~\ref{thm:coset}, making it explicit allows us to directly relate the strong monogamy game from~\cite{coladangelo2021hidden} (which we refer to as the ``coset-monogamy game'') to a simple variant of the monogamy game from~\cite{tomamichel2013monogamy} (which we refer to as the ``basis-monogamy game'') whose maximum success probability we bound using a similar technique to the one introduced in their paper. Ultimately this ``proof by reduction'' is very similar to the direct proof; we include it in the hope that the simple reduction pointed out here will find further uses in the analysis of monogamy games motivated by tasks in quantum cryptography. 

In Section~\ref{sec:coset-game} we introduce the strong monogamy game (called coset-monogamy game here) and state our main result, Theorem~\ref{thm:coset}. In Section~\ref{sec:direct} we prove our main result. In Section~\ref{sec:basis-game} we introduce and analyze our variant of the BB'84-based monogamy game from~\cite{tomamichel2013monogamy} (called basis-monogamy game here). Finally in Section~\ref{sec:reduction} we show a reduction from the coset monogamy game to the basis monogamy game. 

\paragraph{Acknowledgments.} We thank Fatih Kaleoglu for pointing out an error in an earlier proof of Lemma~\ref{lem:permutations}, and an anonymous QIP referee for several corrections.
E.C. would like to thank Anne Broadbent. E.C.'s work is supported by a  CGS M scholarship from Canada's NSERC. T.V.\ is supported by NSF CAREER Grant CCF-1553477, AFOSR YIP award number FA9550-16-1-0495, MURI Grant FA9550-18-1-0161 and the IQIM, an NSF Physics Frontiers Center (NSF Grant PHY-1125565) with support of the Gordon and Betty Moore Foundation (GBMF-12500028).

\section{The coset-monogamy game}
\label{sec:coset-game}

The following game is a monogamy game introduced in~\cite{coladangelo2021hidden}, where it is called ``strong monogamy game'' (see Section 4.4 therein). For a linear subspace $A$ of $\F_2^n$ and $s,s'\in\{0,1\}^n$ recall the notation 
\[ \ket{A} \,=\, \frac{1}{\sqrt{|A|}}\sum_{u\in A} \ket{u}\qquad\text{and}\qquad  \ket{A_{s,s'}}\,=\, X^s Z^{s'} \ket{A}\,=\, \frac{1}{\sqrt{|A|}}\sum_{u\in A} (-1)^{u\cdot s'}\ket{u+s} \;,\]
where $X^s = X^{s_1}\otimes \cdots X^{s_n}$, $Z^{s'}=Z^{s'_1}\otimes \cdots Z^{s'_n}$ with $X = \begin{pmatrix} 0 & 1 \\ 1& 0 \end{pmatrix}$ and $Z = \begin{pmatrix} 1 & 0 \\ 0 & -1 \end{pmatrix}$. 

We formulate the game exactly as in~\cite[Section 4.4]{coladangelo2021hidden}. The only difference is that we rename $\mA_0$ into ``the adversary'', $\mA_1$ into ``Bob'' and $\mA_2$ into ``Charlie''. Thus the game is played between a trusted ``challenger'' and two untrusted, cooperating players Bob and Charlie. The game is parametrized by an even integer $n\geq 2$. 

\bigskip

\underline{Coset-monogamy game.}
\begin{enumerate}
	\item \emph{Preparation:} The challenger picks a uniformly random subspace $A \subseteq \F_2^n$ of dimension $\frac{n}{2}$ and two uniformly random elements $s,s'\in \F_2^n$. The challenger sends $\ket{A_{s,s'}}$ to the adversary. 
	\item The adversary applies a quantum channel $\Phi:\mH_\reg{A} \to \mH_\reg{B} \otimes \mH_{\reg{C}}$, where $\mH_\reg{A} = (\C^2)^{\otimes n}$ and $\mH_\reg{B}$, $\mH_\reg{C}$ are arbitrary. The adversary computes $\rho_{\reg{BC}} = \Phi(\proj{A_{s,s'}})$. It sends registers $\reg{B}$ to Bob and $\reg{C}$ to Charlie, respectively. 
	\item \emph{Question:} The challenger sends the description of $A$, in the form of a basis for it, to both Bob and Charlie.
	\item \emph{Answer:} Bob returns $s_1\in \F_2^n$ and Charlie returns $s_2 \in \F_2^n$. 
	\item \emph{Winning condition:} The adversary, Bob, and Charlie win if and only if $s_1 \in A+s$ and $s_2 \in A^\perp + s'$, where $A^\perp = \{v\in\F_2^n:\, v\cdot u = 0 \,\forall u\in A\}$. 
\end{enumerate}

Our main result is a bound on the maximum winning probability of the adversary, Bob, and Charlie in the coset-monogamy game.

\begin{theorem}\label{thm:coset}
	Let $n\geq 1$ be an even integer. 
	Let $q_n$ be the adversary, Bob, and Charlie's maximum probability of winning in the coset-monogamy game. Then 
	\[q_n \,\leq\, \sqrt{e}\parens*{\cos\frac{\pi}{8}}^n\;.\]
\end{theorem}

\begin{remark}
We have that $\cos\frac{\pi}{8}\approx 0.924$, whereas in~\cite{tomamichel2013monogamy} the bound $(1/2+1/(2\sqrt{2}))^n \approx 0.854^n$ is obtained on the success probability for the variant of the game where Bob and Charlie both have to answer a complete string of measurement outcomes $y,z\in\{0,1\}^n$. Since our version of the game is easier, the bound is slightly weaker. We did not attempt to check if the bound we obtain is optimal. 
\end{remark}

We give two proofs of the theorem. Ultimately, both proofs rely on the technique from~\cite{tomamichel2013monogamy}, and lead to the same numerical bound on the success probability. The difference is that the first proof is direct, while the second proof proceeds by a reduction to a variant of the monogamy game from~\cite{tomamichel2013monogamy}. Since the reduction is intuitively clear, and the monogamy game we reduce to, being based on BB'84 states, is easier to analyze, the second proof is conceptually simpler and potentially more general. However, it is less direct.

\section{Direct proof}
\label{sec:direct}

We give a direct proof of Theorem~\ref{thm:coset}. The proof proceeds in two steps. In the first step we reduce to the analysis of an extended nonlocal game of the form considered in~\cite{JMRW16}. This step is standard in the analysis of monogamy games, and also appears as~\cite[Lemma C.6]{coladangelo2021hidden}. We formulate it in Lemma~\ref{lem:correspondence} below. In the second step we bound the maximum success probability in the extended nonlocal game. This step relies on a technique introduced in~\cite{tomamichel2013monogamy} to bound the operator norm of a tripartite operator introduced to model the players' actions in the game. We describe this step in Section~\ref{sec:tom}.

\subsection{Reduction to an extended nonlocal game}

Write $\textsf{G}\parens*{\tfrac{n}{2},n}$ for the set of linear subspaces of $\F_2^n$ of dimension $\frac{n}{2}$. For $A\in \textsf{G}\parens*{\tfrac{n}{2},n}$ write $\textsf{CS}(A)$ for a fixed set of representatives of the cosets of $A$. In particular, $|\textsf{CS}(A)|=2^{\frac{n}{2}}$.

\begin{lemma}\label{lem:correspondence}
Fix a strategy for the coset-monogamy game, consisting of a channel $\Phi:\mH_\reg{A} \to \mH_\reg{B} \otimes \mH_{\reg{C}}$ and for each $A\in\textsf{G}\parens*{\tfrac{n}{2},n}$ POVMs $\{B^A_s\}_{s\in\textsf{CS}(A)}$ for Bob and  $\{C^A_{s'}\}_{s'\in\textsf{CS}(A^\perp)}$ for Charlie. Let $q'_n$ be the probability that this strategy succeeds in the game. Then
	\begin{align*}
	q_n'&=\Es{A\in\textsf{G}\parens*{\frac{n}{2},n}}\Es{\substack{s\in\textsf{CS}(A)\\s'\in\textsf{CS}(A^\perp)}}\Tr\big( (B^A_s\otimes C^A_{s'})\Phi\big(\proj{A_{s,s'}}\big)\big)\\
	&=\Es{A\in \textsf{G}\parens*{\tfrac{n}{2},n}}\sum_{\substack{s\in\textsf{CS}(A)\\s'\in\textsf{CS}(A^\perp)}}\Tr\big(\big(\proj{A_{s,s'}}\otimes B^A_s\otimes C^A_{s'}\big)\rho\big)\;,
	\end{align*}
	where $\rho=(\Id_{\reg{A}}\otimes\Phi_{\reg{A'}})(\proj{\phi^+}_{\reg{AA'}}^{\otimes n})$ with $\ket{\phi^+}$ the EPR pair, $\ket{\phi^+} = \frac{1}{\sqrt{2}}(\ket{00}+\ket{11})$ and all expectations are uniform averages. 
\end{lemma}

While the first equality is by definition, the second equality is what we refer to as a ``reduction to an extended nonlocal game.'' This is because the second line can be interpreted as the success probability in the following three-player game: (i) Bob and Charlie prepare a tripartite state $\rho_{\reg{ABC}}$ such that $\reg{A}$ is an $n$-qubit register. They give $\reg{A}$ to Alice and keep $\reg{B}$ and $\reg{C}$ respectively. (ii) Alice selects a uniformly random subspace $A \in \textsf{G}\parens*{\tfrac{n}{2},n}$ and gives $A$ to Bob and Charlie. She measures $\reg{A}$ using the projective measurement $\{\proj{A_{s,s'}}\}$ with outcomes $(s,s')\in \textsf{CS}(A)\times \textsf{CS}(A^\perp)$. (iii) Bob and Charlie measure their registers using arbitrary POVM $\{B^A_s\}$ and $\{C^A_{s'}\}$ respectively. They win if and only if they obtain outcomes, $s$ for Bob and $s'$ for Charlie, that match Alice's. 

\begin{proof}
To show the second equality we expand using the definition of $\rho$
	\begin{align*}
	\Tr\big(\big(\proj{A_{s,s'}}\otimes B^A_s\otimes C^A_{s'}\big)\rho\big)&=\frac{1}{2^n}\sum_{r,r'\in\mathbb{F}_2^n}\Tr\big(\big(\proj{A_{s,s'}}\otimes B^A_s\otimes C^A_{s'}\big)\big(\ket{r}\!\bra{r'}\otimes\Phi(\ket{r}\!\bra{r'})\big)\big)\\
	&=\frac{1}{2^n}\sum_{r,r'\in\mathbb{F}_2^n}\langle r' \ket{A_{s,s'}}\langle r\ket{A_{s,s'}}\Tr\big(\big(B^A_s\otimes C^A_{s'}\big)\Phi(\ket{r}\!\bra{r'})\big)\\
	&=\frac{1}{2^n}\Tr\Big( \big(B^A_s\otimes C^A_{s'}\big)\Phi\Big(\sum_{r\in\mathbb{F}_2^n} \ket{r} \langle r \proj{A_{s,s'}}\sum_{r'\in\mathbb{F}_2^n}\proj{r'}\Big)\Big)\\
	&=\frac{1}{2^n}\Tr\big( \big(B^A_s\otimes C^A_{s'}\big)\Phi\big(\proj{A_{s,s'}}\big)\big)\;,
	\end{align*}
	which gives the result.
\end{proof}

\subsection{Analysis of extended nonlocal game}
\label{sec:tom}

We need two preliminary lemmas. The first bounds the overlap of operators constructed as sums of coset state projections. We use $\|\cdot\|$ to denote the operator norm, i.e.\ the largest singular value. 

\begin{lemma}\label{lem:overlaps}
	For any $A,B\in \textsf{G}\parens*{\tfrac{n}{2},n}$, $s'\in\textsf{CS}(A^\perp)$ and $t\in\textsf{CS}(B)$ we have that the overlap
	\begin{align}
	\Big\|\sum_{s\in\textsf{CS}(A)}\proj{A_{s,s'}}\sum_{t'\in\textsf{CS}(B^\perp)}\proj{B_{t,t'}}\Big\|\,\leq\,\sqrt{2^{\dim(A\cap B)-\frac{n}{2}}}\;.
	\end{align}
\end{lemma}

\begin{proof}
	First, note that
	\begin{align}
	\sum_{t'\in\textsf{CS}(B^\perp)}\proj{B_{t,t'}}&=\frac{1}{2^{\frac{n}{2}}}\sum_{t'\in\mathbb{F}_2^n}\proj{B_{t,t'}}\notag\\
	&=\frac{1}{2^n}\sum_{t'\in\mathbb{F}_2^n}\sum_{b,b'\in B}(-1)^{(b+b')\cdot t'}\kb{b+t}{b'+t}\notag\\
	&=\sum_{b,b'\in B}\delta_{b,b'}\kb{b+t}{b'+t}\notag\\
	&=\sum_{b\in B+t}\proj{b}\;,
	\end{align}
 a projection onto the subspace spanned by the vectors given by the elements of the coset $B+t$. Let $\Pi_{B+t} = \sum_{b\in B+t}\proj{b}$. Then
	\begin{align}
\Big\|\sum_{s\in\textsf{CS}(A)}\proj{A_{s,s'}}\sum_{t'\in\textsf{CS}(B^\perp)}\proj{B_{t,t'}}\Big\|
&= \Big\| \sum_{s\in\textsf{CS}(A)}\proj{A_{s,s'}} \Pi_{B+t} \Big\|\notag\\
&= \Big\|  \Pi_{B+t}  \Big( \sum_{s\in\textsf{CS}(A)}  \proj{A_{s,s'}}\Big) \Pi_{B+t} \Big\|^{1/2}
\;,\label{eq:overlaps-1a}
\end{align}
where the second equality uses that $\{\proj{A_{s,s'}}\}_{s\in \textsf{CS}(A)}$ are orthogonal projectors. Since we have that $\ket{A_{s,s'}}$ is a superposition of basis elements in $A+s$, so $\Pi_{B+s}\ket{A_{s,s'}}$ is a superposition of basis elements in $(A+s)\cap(B+t)$, giving that the set of $\Pi_{B+t}\ket{A_{s,s'}}$ over $s$ is orthogonal. Thus,
\begin{align}
\Big\|  \Pi_{B+t}  \Big( \sum_{s\in\textsf{CS}(A)}  \proj{A_{s,s'}}\Big) \Pi_{B+t} \Big\|&\leq \max_{s\in\textsf{CS}(A)} \big\| \Pi_{B+t}  \proj{A_{s,s'}}\Pi_{B+t}\big\| \notag\\
&=\max_{s\in\textsf{CS}(A)} \bra{A_{s,s'}} \Pi_{B+t} \ket{A_{s,s'}}
\;,\label{eq:overlaps-1}
\end{align}
which uses $\|\sum_s X_s\|\leq \max_s \|X_s\|$ for $X_i$ Hermitian with orthogonal range.
 Now, for any $s\in \textsf{CS}(A)$,
	\begin{align*}
	 \bra{A_{s,s'}} \Pi_{B+t} \ket{A_{s,s'}} &= \frac{1}{2^{\frac{n}{2}}}\big|(A+s)\cap(B+t)\big|\\
	&\leq\frac{1}{2^{\frac{n}{2}}}\big|A\cap B\big|\;.
	\end{align*}
	Plugging this back into~\eqref{eq:overlaps-1} completes the proof.
\end{proof}

The second lemma is a key bound used in~\cite{tomamichel2013monogamy}.

\begin{lemma}[Lemma 2 in~\cite{tomamichel2013monogamy}]\label{lem:sum-bound}
	Let $P_1,...,P_n$ be positive semidefinite operators on a Hilbert space. Then
	\begin{align*}
		\Big\|\sum_{i=1}^nP_i\Big\| &\leq\sum_{i=1}^n\,\max_{j=1,...,n}\,\big\|\sqrt{P_j}\sqrt{P_{\pi_i(j)}}\big\|\;,
	\end{align*}
	where $\pi_1,...,\pi_n$ is any set of mutually orthogonal permutations of $\{1,...,n\}$, i.e. $\pi_i\circ\pi_j^{-1}$ only has a fixed point if $i=j$.
\end{lemma}

We give the permutations we will use to apply Lemma~\ref{lem:sum-bound}. For $n\geq 2$ even let 
\[C_{n,n/2} \,=\,  \big\{\gamma\in \{0,1\}^n:\, |\gamma|=\frac{n}{2}\big\}\;,\]
 where for a string $\gamma$, $|\gamma|$ denotes its Hamming weight (number of nonzero entries).

\begin{lemma}\label{lem:permutations}
	Let $n$ be an even integer. Then there are $N = {n \choose n/2}$ mutually orthogonal permutations $\pi_1,\ldots,\pi_N$ of $C_{n,n/2}$
	such that the following holds. For each $k\in\{0,\ldots,\frac{n}{2}\}$ 
	there are exactly $\binom{\frac{n}{2}}{k}^2$ permutations $\pi_j$ such that the number of positions at which $\gamma$ and $\pi_j(\gamma)$ are both $1$ is $\frac{n}{2}-k$. 
\end{lemma}

\begin{proof}
Fix $n\geq 2$ an even integer, and let $k\in \{1,\ldots,\frac{n}{2}\}$. Let $G_{n,k}$ be the graph with vertex set $C_{n,n/2}$ and an edge between any $\gamma,\gamma'\in C_{n,n/2}$ such that the number of positions at which $\gamma$ and $\gamma'$ are both $1$ is exactly $\frac{n}{2}-k$.

We claim that the minimum degree $d_k$ of $G_{n,k}$ is at least  $\binom{\frac{n}{2}}{k}^2$. Indeed, for any $\gamma\in C_{n,n/2}$ we can define distinct $\gamma'$ that are connected to it in $G_{n,k}$ by choosing $k$ locations among the $\frac{n}{2}$ $1$ positions of $\gamma$, $k$ locations among the $\frac{n}{2}$ $0$ positions, and flipping those values. 

For each edge in $G_{n,k}$ create two directed edges to obtain a directed graph $\tilde{G}_{n,k}$. In $\tilde{G}_{n,k}$ each vertex has in-degree at least $d_k$, and out-degree at least $d_k$. Thus we can find $d_k$ non-overlapping oriented vertex cycle covers of $\tilde{G}_{n,k}$, call them $c_{k,1},\ldots,c_{k,d_k}$.\footnote{To show this, find a first cycle cover in an arbitrary way and remove all edges used. This reduces both the out- and in-degrees by exactly $1$. Repeat until the minimum degree reaches zero.} To each such oriented vertex cycle cover associate a permutation $\pi_{k,i}$ of $C_{n,n/2}$ in the natural way. By construction for any $i\neq i'$, $\pi_{k,i}$ and $\pi_{k,i'}$ are orthogonal. 

For $k=0$, set $\pi_{0,1}$ to be the identity permutation of $C_{n,n/2}$. 

We observe that for $k\neq k'$ and any $i,i'$ it must be that $\pi_{k,i}$ and $\pi_{k',i'}$ are orthogonal permutations. This is because two elements of $C_{n,n/2}$ can be connected by an edge in at most one $G_{n,k}$. To conclude, use that by the Vandermonde identity we have found a total of 
	\begin{equation*}
	N \,=\, {n\choose n/2} \,=\, \sum_{k=0}^{n/2}{n/2 \choose k}^2
	\end{equation*}
	mutually orthogonal permutations, as desired. 
\end{proof}

\begin{figure}
	\centering
	\begin{subfigure}[b]{0.3\textwidth}
		\centering	
		\begin{tikzpicture}
		
		\draw (0,0) -- (0.866,0.5) -- (0.866,1.5) -- (-0.866, 1.5) -- (-0.866,0.5) -- (0.866,0.5) -- (0,2) -- (-0.866,1.5) -- (0,0) -- (0.866,1.5) -- (-0.866,0.5) -- (0,2) -- (0,0);
		
		\fill (0,0) circle (2pt) node[below]{$0011$};
		\fill (0.866,0.5) circle (2pt) node[below right]{$0101$};
		\fill (0.866,1.5) circle (2pt) node[above right]{$0110$};
		\fill (0,2) circle (2pt) node[above]{$1001$};
		\fill (-0.866,1.5) circle (2pt) node[above left]{$1010$};
		\fill (-0.866,0.5) circle (2pt) node[below left]{$1100$};
		\end{tikzpicture}
		\caption{The graph $G_{4,1}$.}
	\end{subfigure}
	\hfill
	\begin{subfigure}[b]{0.3\textwidth}
		\centering	
		\begin{tikzpicture}[scale=1.3]
		
		\begin{scope}[decoration={
			markings,
			mark=at position 0.5 with {\arrow{latex}}}
		] 
		\draw[postaction=decorate] (0,0) to [bend right=15] (0.866,0.5);
		\draw[postaction=decorate] (0.866,0.5) to [bend right=15] (0,0);
		\draw[postaction=decorate] (0,0) to [bend right=15] (0.866,1.5);
		\draw[postaction=decorate] (0.866,1.5) to [bend right=15] (0,0);
		\draw[postaction=decorate] (0,0) to [bend right=15] (0,2);
		\draw[postaction=decorate] (0,2) to [bend right=15] (0,0);
		\draw[postaction=decorate] (0,0) to [bend right=15] (-0.866,1.5);
		\draw[postaction=decorate] (-0.866,1.5) to [bend right=15] (0,0);
		\draw[postaction=decorate] (0.866,0.5) to [bend right=15] (-0.866,0.5);
		\draw[postaction=decorate] (-0.866,0.5) to [bend right=15] (0.866,0.5);
		\draw[postaction=decorate] (0.866,0.5) to [bend right=15] (0,2);
		\draw[postaction=decorate] (0,2) to [bend right=15] (0.866,0.5);
		\draw[postaction=decorate] (0.866,0.5) to [bend right=15] (0.866,1.5);
		\draw[postaction=decorate] (0.866,1.5) to [bend right=15] (0.866,0.5);
		\draw[postaction=decorate] (0.866,1.5) to [bend right=15] (-0.866,0.5);
		\draw[postaction=decorate] (-0.866,0.5) to [bend right=15] (0.866,1.5);
		\draw[postaction=decorate] (0.866,1.5) to [bend right=15] (-0.866,1.5);
		\draw[postaction=decorate] (-0.866,1.5) to [bend right=15] (0.866,1.5);
		\draw[postaction=decorate] (0,2) to [bend right=15] (-0.866,0.5);
		\draw[postaction=decorate] (-0.866,0.5) to [bend right=15] (0,2);
		\draw[postaction=decorate] (0,2) to [bend right=15] (-0.866,1.5);
		\draw[postaction=decorate] (-0.866,1.5) to [bend right=15] (0,2);
		\draw[postaction=decorate] (-0.866,0.5) to [bend right=15] (-0.866,1.5);
		\draw[postaction=decorate] (-0.866,1.5) to [bend right=15] (-0.866,0.5);
		\end{scope}
		
		\end{tikzpicture}
		\caption{The graph $\tilde{G}_{4,1}$.}
	\end{subfigure}
\hfill
\begin{subfigure}[b]{0.3\textwidth}
\centering	
\begin{tikzpicture}[scale=0.6]

	\begin{scope}[decoration={
		markings,
		mark=at position 0.7 with {\arrow{latex}}}
	]
		\draw[postaction=decorate] (2.5,0) -- (3.366,0.5);
		\draw[postaction=decorate] (3.366,0.5) -- (1.634,0.5);
		\draw[postaction=decorate] (3.366,1.5) -- (2.5,0);
		\draw[postaction=decorate] (2.5,2) -- (1.634,1.5);
		\draw[postaction=decorate] (1.634,1.5) -- (3.366,1.5);
		\draw[postaction=decorate] (1.634,0.5) -- (2.5,2);
	\end{scope}
	
	\begin{scope}[decoration={
		markings,
		mark=at position 0.7 with {\arrow{latex}}}
	]
		\draw[postaction=decorate] (0,2.5) -- (0.866,4);
		\draw[postaction=decorate] (0.866,3) -- (0,2.5);
		\draw[postaction=decorate] (0.866,4) -- (0.866,3);
		\draw[postaction=decorate] (0,4.5) -- (-0.866,3);
		\draw[postaction=decorate] (-0.866,4) -- (0,4.5);
		\draw[postaction=decorate] (-0.866,3) -- (-0.866,4);
	\end{scope}
	
	\begin{scope}[decoration={
		markings,
		mark=at position 0.7 with {\arrow{latex}}}
	]
		\draw[postaction=decorate] (2.5,2.5) to [bend right=15] (2.5,4.5);
		\draw[postaction=decorate] (2.5,4.5) to [bend right=15] (2.5,2.5);
		\draw[postaction=decorate] (3.366,3) -- (3.366,4);
		\draw[postaction=decorate] (3.366,4) -- (1.634,4);
		\draw[postaction=decorate] (1.634,4) -- (1.634,3);
		\draw[postaction=decorate] (1.634,3) -- (3.366,3);
	\end{scope}
	
	\begin{scope}[decoration={
		markings,
		mark=at position 0.7 with {\arrow{latex}}}
	] 
	
		\draw[postaction=decorate] (0.866,1.5) to [bend right=15] (-0.866,0.5);
		\draw[postaction=decorate] (0,2) to [bend right=15] (0.866,0.5);
		\draw[postaction=decorate] (0,0) to [bend right=15] (-0.866,1.5);
		\draw[postaction=decorate] (0.866,0.5) to [bend right=15] (0,2);
		\draw[postaction=decorate] (-0.866,0.5) to [bend right=15] (0.866,1.5);
		\draw[postaction=decorate] (-0.866,1.5) to [bend right=15] (0,0);
	\end{scope}
		\end{tikzpicture}
		\caption{A set of cycle covers.}
	\end{subfigure}
	\caption{Graph construction from Lemma~\ref{lem:permutations} in the case of $n=4,k=1$.}
	\label{fig:graphconstruction}
\end{figure}
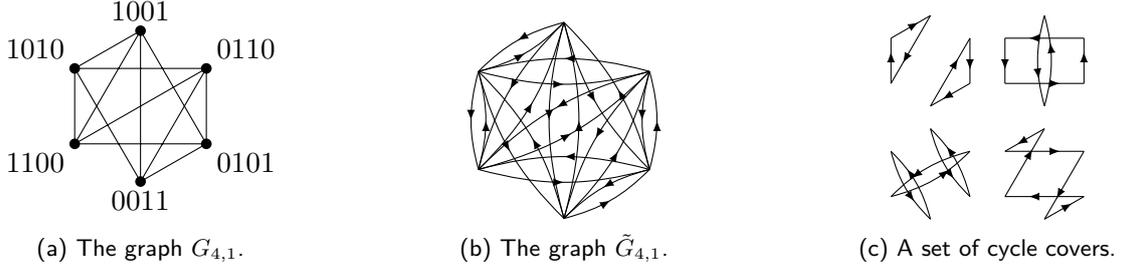

\begin{remark}
For any set $X$ of cardinality $n$, we can consider the $\pi_j$ as permutations on the collection of subsets of size $\tfrac{n}{2}$, instead of permutations of $C_{n,n/2}$. We do this by fixing an ordering of the elements of $X$, and referring to the subsets by their indicator strings. The set of indicator strings is $C_{n,n/2}$, on which $\pi_j$ acts. Below, we apply this remark where the set $X$ in question is a basis of $\F_2^n$, and the subsets are bases of subspaces of dimension $\tfrac{n}{2}$.
\end{remark}

We are ready to complete our proof of the upper bound on the winning probability of the coset-monogamy game.

\begin{proof}[Proof of Theorem~\ref{thm:coset}]
Fix a strategy for the coset-monogamy game, consisting of a channel $\Phi:\mH_\reg{A} \to \mH_\reg{B} \otimes \mH_{\reg{C}}$ and, for each $A\in\textsf{G}\parens*{\tfrac{n}{2},n}$, POVMs $\{B^A_s\}_{s\in\textsf{CS}(A)}$ for Bob and  $\{C^A_{s'}\}_{s'\in\textsf{CS}(A^\perp)}$ for Charlie. Let $q'_n$ be the probability that this strategy succeeds in the game. Using Naimark's theorem as in Lemma 9 of \cite{tomamichel2013monogamy}, we may assume that the POVMs are projective. Using Lemma~\ref{lem:correspondence},
	\begin{align*}
	q_n'&=\Es{A\in \textsf{G}\parens*{\tfrac{n}{2},n}}\sum_{\substack{s\in\textsf{CS}(A)\\s'\in\textsf{CS}(A^\perp)}}\Tr\big(\big(\proj{A_{s,s'}}\otimes B^A_s\otimes C^A_{s'}\big)\rho\big)\\
	&\leq\Big\|\Es{A\in \textsf{G}\parens*{\tfrac{n}{2},n}}\Pi^A\Big\|\;,
	\end{align*}
	where 
	\[\Pi^A=\sum_{\substack{s\in\textsf{CS}(A)\\s'\in\textsf{CS}(A^\perp)}}\proj{A_{s,s'}}\otimes B^A_s\otimes C^A_{s'}\;.\]
	As in \cite{coladangelo2021hidden} we  decompose the average over the subspaces followed by an average over bases of $\mathbb{F}_2^n$, and then over subspaces that may be spanned by $\frac{n}{2}$ vectors from the basis. Using the triangle inequality we can bound the winning probability as
	\begin{align}
	q_n'&\leq\Es{\beta\text{ basis of }\mathbb{F}_2^n}\Big\|\Es{\substack{\gamma\subseteq\beta\\|\gamma|=\frac{n}{2}}}\Pi^{\mathrm{span}(\gamma)}\Big\|\;.
	\end{align}
	We apply Lemma~\ref{lem:sum-bound} using the permutations $\pi_1,\ldots,\pi_N$ from Lemma~\ref{lem:permutations}, where $N={n\choose n/2}$. Applying the lemma, 
		\begin{align}
		q_n' &\leq\Es{\beta\text{ basis of }\mathbb{F}_2^n}\frac{1}{N}\sum_{j=1}^N\max_{\substack{\gamma\subseteq\beta\\|\gamma|=\frac{n}{2}}}\big\|\Pi^{\mathrm{span}(\gamma)}\Pi^{\mathrm{span}(\pi_{j}(\gamma))}\big\|\;.\label{eq:cosett-1}
	\end{align}
	For any subspaces $A,B\in \textsf{G}\parens*{\tfrac{n}{2},n}$ define the projectors
	\begin{align}
	P=\sum_{\substack{s\in\textsf{CS}(A)\\s'\in\textsf{CS}(A^\perp)}}\proj{A_{s,s'}}\otimes\Id_\textsf{B}\otimes C^A_{s'}\qquad\text{and}\qquad Q=\sum_{\substack{s\in\textsf{CS}(B)\\s'\in\textsf{CS}(B^\perp)}}\proj{B_{s,s'}}\otimes B^B_s\otimes \Id_\textsf{C}\;,
	\end{align}
	which satisfy $\Pi^A\leq P$ and $\Pi^B\leq Q$. Thus
	\begin{align}
	\norm{\Pi^A\Pi^B}^2=\sup_{\ket{v}}\bra{v}\Pi^B\Pi^A\Pi^B\ket{v}=\sup_{\ket{v}\in\supp(\Pi_B)}\bra{v}\Pi^A\ket{v}\leq\sup_{\ket{v}\in\supp (Q)}\bra{v}P\ket{v}=\norm{PQ}^2\;,
	\end{align}
	and using Lemma~\ref{lem:overlaps},
	\begin{align}
	\norm{\Pi^A\Pi^B}&\leq\Big\|\sum_{\substack{(s,s')\in\textsf{CS}(A)\times\textsf{CS}(A^\perp)\\(t,t')\in\textsf{CS}(B)\times\textsf{CS}(B^\perp)}}\ket{A_{s,s'}}\!\bra{A_{s,s'}}\cdot\ket{B_{t,t'}}\!\bra{B_{t,t'}}\otimes B^B_t\otimes C^A_{s'}\Big\|\notag\\
	&=\max_{\substack{s'\in\textsf{CS}(A^\perp)\\t\in\textsf{CS}(B)}}\Big\|\sum_{s\in\textsf{CS}(A)}\proj{A_{s,s'}}\sum_{t'\in\textsf{CS}(B^\perp)}\proj{B_{t,t'}}\Big\|\notag\\
	&\leq \sqrt{2^{\dim(A\cap B)-\frac{n}{2}}}\;.\label{eq:cosett-2}
	\end{align}
	By Lemma~\ref{lem:permutations} for $k\in \{0,\ldots,\frac{n}{2}\}$ there are $\binom{n/2}{k}^2$ permutations $\pi_j$ such that the dimension of $\mathrm{span}(\gamma)\cap\mathrm{span}(\pi_{j}(\gamma))$ is $\frac{n}{2}-k$.  Plugging~\eqref{eq:cosett-2} back into~\eqref{eq:cosett-1} we thus get
	\begin{align*}
	q_n'&\leq\Es{\beta\text{ basis of }\mathbb{F}_2^n}\frac{1}{N}\sum_{j=1}^{N}\max_{\substack{\gamma\subseteq\beta\\|\gamma|=\frac{n}{2}}}\sqrt{2^{\dim(\mathrm{span}(\gamma)\cap\mathrm{span}(\pi_{j}(\gamma)))-\frac{n}{2}}}\\
	&\leq\frac{1}{\binom{n}{n/2}}\sum_{k=0}^{n/2}\binom{n/2}{k}^2\sqrt{2^{-k}}\;.
	\end{align*}
	The final bound is provided by Lemma~\ref{lem:binomial} stated below. 
	\end{proof}

\begin{lemma}\label{lem:binomial}
For any even integer $n\geq 2$,
\[ \frac{1}{\binom{n}{n/2}}\sum_{k=0}^{n/2}\binom{n/2}{k}^2\sqrt{2^{-k}}\,\leq\, \sqrt{e}\parens*{\cos\frac{\pi}{8}}^n\;.  \]
\end{lemma}

\begin{proof}
We bound $\binom{n/2}{k}\leq\binom{n/2}{n/4}$ for any $k\in \{0,\ldots,\frac{n}{2}\}$ and 
	\begin{align*}
	\frac{\binom{n/2}{n/4}}{\binom{n}{n/2}}&=\frac{1}{2^\frac{n}{2}}\parens*{1+\frac{1}{n-1}}\parens*{1+\frac{1}{n-3}}\cdots \parens*{1+\frac{1}{\frac{n}{2}+1}}\leq\frac{1}{2^\frac{n}{2}}\parens*{1+\frac{1}{\frac{n}{2}+1}}^{\frac{n}{4}}\\
	&\leq\frac{\sqrt{e}}{2^\frac{n}{2}}\;,
	\end{align*}
	which gives
	\begin{align*}
	\frac{1}{\binom{n}{n/2}}\sum_{k=0}^{n/2}\binom{n/2}{k}^2\sqrt{2^{-k}}&\leq\frac{\binom{n/2}{n/4}}{\binom{n}{n/2}}\sum_{k=0}^{n/2}\binom{n/2}{k}\sqrt{2^{-k}}\\
	&\leq\frac{\sqrt{e}}{2^\frac{n}{2}}\parens*{1+\frac{1}{\sqrt{2}}}^\frac{n}{2}\\
	&=\sqrt{e}\parens*{\cos\frac{\pi}{8}}^n\;,
	\end{align*}
	as claimed.
\end{proof}

\section{The basis-monogamy game}
\label{sec:basis-game}

In this section we introduce a monogamy game which we call the \emph{basis-monogamy game}. While this game is conceptually simpler than the coset-monogamy game introduced in Section~\ref{sec:coset-game}, in the next section we will show that the latter can be reduced to the former. Here we focus on the basis-monogamy game, which may be of independent interest, and its analysis. 

We formulate the game directly as an extended nonlocal game, that can be seen as a variant of a game introduced in~\cite{tomamichel2013monogamy}. Informally, in the game from~\cite{tomamichel2013monogamy} two players Bob and Charlie are trying to both be maximally entangled with Alice: they are required to prepare a tripartite state $\rho_{\reg{ABC}}$, where $\reg{A}$ is an $n$-qubit register handed over to Alice, and $\reg{B}$ and $\reg{C}$ are arbitrary registers kept by Bob and Charlie respectively, such that when Alice measures her $n$ qubits in a randomly chosen basis $\theta\in\{0,1\}^n$ (where as usual $\theta_i=0$ denotes a measurement in the standard basis, and $\theta_i=1$ a measurement in the Hadamard basis) to obtain a string of outcomes $x\in\{0,1\}^n$, given $\theta$ as side information Bob and Charlie are able to return strings $y,z\in\{0,1\}^n$ respectively such that $x=y=z$. Our variant of the game introduces two simple modifications: first, $n$ is even and $\theta$ is chosen such that $|\theta| =  \frac{n}{2}$, and second, Bob and Charlie are only asked to predict measurement outcomes associated with the standard basis ($\theta_i=0$) and Hadamard basis ($\theta_i=1$), respectively. More formally, for $n$ an even integer the basis-monogamy game proceeds as follows.  

\bigskip

\underline{Basis-monogamy game.}
\begin{enumerate}
\item \emph{Preparation:} Bob and Charlie together prepare a state $\rho_{\reg{ABC}}$ such that $\reg{A}$ is an $n$-qubit register and $\reg{B}$ and $\reg{C}$ are arbitrary. They pass $\reg{A}$ to Alice and keep registers $\reg{B}$ and $\reg{C}$ to themselves, respectively. 
\item \emph{Question:} Alice chooses $\theta \in \{0,1\}^n$ uniformly at random conditioned on $|\theta|=\frac{n}{2}$. Alice measures each qubit of $\reg{A}$ in the basis indicated by $\theta$ to obtain a string of outcomes $x\in\{0,1\}^n$. She sends $\theta$ to Bob and Charlie. Let $T = \{ i\in\{1,\ldots,n\}:\, \theta_i = 0\}$. 
\item \emph{Answer:} Bob returns a string $y\in \{0,1\}^T$. Charlie returns a string $z \in \{0,1\}^{\ol{T}}$. 
\item \emph{Winning condition:} Bob and Charlie win if and only if $y=x_T$ and $z=x_{\ol{T}}$. 
\end{enumerate}

Naturally this game is slightly easier than the one considered in~\cite{tomamichel2013monogamy}. Nevertheless we can use the same proof technique to bound the maximum success probability and obtain the following result.

\begin{theorem}\label{thm:basis-game}
Let $n\geq 1$ be an even integer. 
Let $p_n$ be Bob and Charlie's maximum probability of winning in the basis-monogamy game. Then 
\[p_n \,\leq\, \sqrt{e}\parens*{\cos\frac{\pi}{8}}^n\;.\]
\end{theorem}
 
\begin{proof}
The proof follows very closely the proof of~\cite[Theorem 3]{tomamichel2013monogamy}.  Fix an arbitrary strategy for the game that succeeds with probability $p'_n$. The strategy consists of a state $\rho_{\reg{ABC}}$ and for each $\theta \in C_{n,n/2} = \{\gamma\in \{0,1\}^n:\, |\gamma|=\frac{n}{2}\}$  two POVMs $\{B^\theta_y\}_{y\in \{0,1\}^T}$ and  $\{C^\theta_z\}_{z\in \{0,1\}^{\ol{T}}}$ respectively. Applying Naimark's dilation theorem if needed, assume without loss of generality that both families of measurements are projective. For any $\theta\in\{0,1\}^n$ such that $|\theta|=\frac{n}{2}$ define
\[ \Pi^\theta \,=\, \sum_{x\in\{0,1\}^n} \proj{x}_\theta \otimes B^\theta_{x_T} \otimes C^\theta_{x_{\ol{T}}}\;.\]
Then $\Pi^\theta$ is a projector. Furthermore we can express the strategy's success probability as
\begin{align}
p'_n&= \Es{ \theta\in C_{n,n/2}} \Tr\big( \Pi^\theta\, \rho_{\reg{ABC}}\big)\notag\\
&\leq \Big\| \Es{\theta\in C_{n,n/2}} \Pi^\theta\Big\|\notag\\
&\leq \frac{1}{N} \sum_{k=1}^N \max_\theta \big\|\Pi^\theta \Pi^{\pi^k(\theta)}\big\|\;,\label{eq:l1-1}
\end{align}
where the first inequality follows by linearity and the definition of the operator norm and the second inequality follows from~\ref{lem:sum-bound}. In the third line we set $N={n \choose n/2}$ and $\pi^1,\ldots,\pi^N$ are the $N$ mutually orthogonal permutations promised by Lemma~\ref{lem:permutations}.

Note that at this stage we are in a situation that is very similar to the situation at Eq.~\eqref{eq:cosett-1} in the proof of Theorem~\ref{thm:coset}. The only difference is that there is a single basis $\beta$, that is the standard basis of $\F_2^n$ (i.e.\ the coordinate vectors). We make the correspondence between the two situations more explicit in Section~\ref{sec:reduction}. Here, for clarity we complete the proof without at all resorting to the notation of subspaces. 

Fix an arbitrary pair $(\theta,\theta')$ and let $R$ be the set of indices in which $\theta$ and $\theta'$ differ. Without loss of generality, assume that $\theta_R$ has Hamming weight at most $|R|/2$; if not we exchange the roles of $\theta$ and $\theta'$. Let $S = \{i\in R:\, \theta_i=0\}$, so that $S\subseteq R$ and $|S|> |R|/2$. Let 
\[ T=\{i:\, \theta_i=0\}\qquad\text{and}\qquad T'=\{i:\, \theta'_i=0\}\;,\]
so that $S\subseteq T\cap \ol{T'}$.
Let 
\[\ol{P} \,=\, \sum_{x_T\in\{0,1\}^T} \big(\proj{x_S} \otimes \Id_{\ol{S}}\big) \otimes B^\theta_{x_T} \otimes \Id_\reg{C}\;,\] 
where $\Id_{\ol{S}} $ denotes the identity on qubits of register $\reg{A}$ that do not lie in the set $S$. 
 Similarly, let 
\[\ol{Q} \,=\, \sum_{x_{\ol{T'}}\in\{0,1\}^{\ol{T'}}} \big(H^{S}\proj{x_S}H^S \otimes \Id_{\ol{S}} \big)\otimes  \Id_\reg{B} \otimes C^{\theta'}_{x_{\ol{T'}}} \;,\]
where   $H^S$ denotes a Hadamard on each of the qubits in $S$.  We  compute 
\begin{align*}
\ol{P}\,\ol{Q}\,\ol{P} &= \sum_{x_T,y_{\ol{T'}},z_T} \proj{x_S} H^S \proj{y_{S}} H^S \proj{z_S} \otimes  \Id_{\ol{S}} \otimes P^\theta_{x_T} P^\theta_{z_T} \otimes Q^{\theta'}_{y_{\ol{T'}}}\\
&= \sum_{x_T,y_{\ol{T'}}} \proj{x_S} H^S \proj{y_{S}} H^S \proj{x_S} \otimes  \Id_{\ol{S}} \otimes P^\theta_{x_T}  \otimes Q^{\theta'}_{y_{\ol{T'}}}\\
&= 2^{-|S|} \sum_{x_T} \proj{x_S}\otimes  \Id_{\ol{S}} \otimes P^\theta_{x_T}  \otimes \Id_\reg{C}\;,
\end{align*}
where for the second line we used that $P^\theta_{x_T} P^\theta_{z_T} = \delta_{x_T,z_T} P^\theta_{x_T}$ and for the third line that $|\bra{x_S} H^S \ket{y_{S}}|^2=2^{-|S|}$ for all $x,y$ and $\sum_{y_{\ol{T'}}} Q^{\theta'}_{y_{\ol{T'}}}=\Id_\reg{C}$ for all $\theta'$. Using that $\sum_{x_T} P^\theta_{x_T} =\Id$ it follows that 
\[\|\ol{P}\ol{Q}\ol{P}\|\leq 2^{-|S|}\,\leq\,2^{-|R|/2}\;,\]
where the second inequality is because $|S|\geq|R|/2$. Hence for all $(\theta,\theta')$,
\begin{align}
\big\| \Pi^\theta \Pi^{\theta'}\big\|^2 &= \big\| \Pi^{\theta'} \Pi^\theta \Pi^{\theta'}\big\|\notag\\
&\leq   \big\|\Pi^{\theta'}\, \ol{P} \,\Pi^{\theta'}\big\|\notag\\
&=  \big\|\ol{P} \,\Pi^{\theta'} \,\ol{P} \big\|\notag\\
&\leq \big\|\ol{P}\,\ol{Q}\,\ol{P}\big\|\notag\\
&\leq 2^{-|R|/2}\;,\label{eq:l1-2}
\end{align}
where in the first equality we used that $\Pi^{\theta}$ is a projection, the first inequality uses $\Pi^\theta \leq \ol{P}$ because $C^\theta_{x_{\ol{T}}}\leq \Id$ for all $x_{\ol{T}}$, the second equality uses that $\ol{P}$ and $\Pi^{\theta'}$ are projections and the last inequality that $\Pi^{\theta'}\leq \ol{Q}$.

By Lemma~\ref{lem:permutations} for any $k\in \{0,\ldots,\frac{n}{2}\}$ there are ${n/2 \choose k}^2$ permutations $\pi_j$ such that $\theta$ and $\pi_j(\theta)$ differ in $2k$ positions, i.e.\ such that $|R|=2k$. 
Returning to~\eqref{eq:l1-1} and using~\eqref{eq:l1-2} we obtain
\begin{align}
p'_N &\leq \frac{1}{N} \sum_{k=0}^{n/2} {n/2 \choose k}^2 2^{-k/2}\;.\label{eq:l1-4}
\end{align}
We conclude using Lemma~\ref{lem:binomial}. 
\end{proof}

\section{Reduction to the coset-monogamy game}
\label{sec:reduction}

In this section we show a reduction from the coset-monogamy game to the basis-monogamy game. This gives a second proof of Theorem~\ref{thm:coset}, by reduction to Theorem~\ref{thm:basis-game}. 

\begin{proposition}\label{prop:coset-game}
Let $n\geq 2$ be an even integer. Let $p_n$ be the maximum probability of winning for Bob and Charlie in the basis-monogamy game. Let $q_n$ be the maximum probability of winning for the adversary, Bob and Charlie in the coset-monogamy game. Then 
\[ q_n \,\leq\, p_n\;.\]
\end{proposition}

\begin{proof}
Let $n\geq 2$ be even. Fix a strategy for the adversary that succeeds with some probability $q'_n \leq q_n$ in the coset-monogamy game. This strategy is specified by a channel $\Phi$ and families of POVM $\{B^A_s\}_{s\in \F_2^n}$ and $\{C^A_s\}_{s\in \F_2^n}$ for Bob and Charlie respectively. Here, the POVMs are indexed by subspaces $A$ and return outcomes $s\in \F_2^n$. 

We define a strategy for Bob and Charlie in the basis-monogamy game that succeeds with probability $p'_n = q'_n$. The strategy is as follows: 
\begin{enumerate}
\item Bob and Charlie prepare $n$ EPR pairs, $\rho_{\reg{AA'}} = \proj{\phi^+}^{\otimes n}$ where  $\ket{\phi^+} = \frac{1}{\sqrt{2}}(\ket{00}+\ket{11})$ and registers $\reg{A}$ and $\reg{A'}$ are $n$ qubits each, containing the $n$ first halves and the $n$ second halves of the EPR pairs respectively. They select a uniformly random basis $\mB=\{u_1,\ldots,u_n\}$ of $\F_2^n$ which they each keep a copy of. Let 
$U_\mB$ be the unitary of $(\C^2)^{\otimes n}$ which permutes standard basis vectors as
\begin{equation}\label{eq:def-ub}
\forall x \in \{0,1\}^n\;,\qquad U_\mB\ket{x}\,=\,\Big|\sum_i x_i u_i\Big\rangle\;.
\end{equation}
They apply $U_B$ to register $\reg{A'}$ and then compute
\begin{equation}\label{eq:rho-abc}
 \rho_{\reg{ABC}} \,=\,\big(\Id_{\reg{A}}\otimes\Phi_{\reg{A'}\to \reg{BC}}\big)  \big( (\Id_\reg{A}\otimes U_\mB)[\rho_{\reg{AA'}}]\big)\;,
\end{equation}
where for any linear maps $X,Y$ on $\mH$ we write $X[Y]$ for $XYX^\dagger$. 
They send register $\reg{A}$ to the challenger. Bob keeps register $\reg{B}$ and Charlie keeps register $\reg{C}$. 
\item Let $\theta\in \{0,1\}^n$ be the question selected by the challenger. Upon receipt of $\theta$, Bob and Charlie each set 
\begin{equation}\label{eq:def-A}
A = \textrm{Span}\{ u_i : \theta_i = 1\}
\end{equation}
and $T = \{i: \theta_i=0\}.$
Bob measures the qubits in $\reg{B}$ using $\{B^A_s\}$ to obtain an outcome $s_1$. Charlie  measures the qubits in $\reg{C}$ using $\{C^A_s\}$ to obtain an outcome $s_2$. Let $s_1 = \sum_i y_i u_i$ and $s_2 = \sum_i z_i u_i$, where $y,z\in \{0,1\}^n$, be the unique decomposition of each vector in the basis $\{u_i\}$ of $\F_2^n$.  Bob returns $y_T$ and Charlie returns $z_{\ol{T}}$. 
\end{enumerate}
We introduce notation to express the winning probability of this strategy in the basis-monogamy game. For a basis $\mB = \{u_1,\ldots,u_n\}$ of $\F_2^n$ and $\theta\in\{0,1\}^n$, $T=\{i:\theta_i=0\}$ and $y\in \{0,1\}^T,z\in\{0,1\}^{\ol{T}}$ we let $A$ be the space spanned by $\{u_i: \theta_i=1\}$ and
\begin{equation}\label{eq:def-BC}
 B^{(\mB,\theta)}_y \,=\, \sum_{y':\,y'\cdot u_i=y_i \forall i\in T} B^A_{y'}\;,\qquad C^{(\mB,\theta)}_z \,=\, \sum_{z':\,z'\cdot u_i=z_i \forall i\in \ol{T}} C^A_{z'}\;.
\end{equation}
Note that here $T$ is determined by $\theta$, and for all $(\mB,\theta)$, both $\{B^{(\mB,\theta)}_y\}_{y\in\{0,1\}^T}$ and $\{C^{(\mB,\theta)}_z\}_{z\in\{0,1\}^{\ol{T}}}$ is a POVM. With this notation we can write
\[ p'_n \,=\, \Es{\theta \in \{0,1\}^n }\Es{\mB} \sum_{x\in \{0,1\}^n} \Tr\Big( \big(\proj{x}_\theta \otimes B^{(\mB,\theta)}_{x_T} \otimes C^{(\mB,\theta)}_{x_{\ol{T}}}\big)\rho_{\reg{ABC}}\Big) \;,\]
where $\rho_{\reg{ABC}}$ is the state defined in~\eqref{eq:rho-abc} and the expectation is over a uniformly random $\theta\in\{0,1\}^n$ (as chosen by the challenger) and basis $\mB=\{u_1,\ldots,u_n\}$ for $\F_2^n$ (as chosen by Bob and Charlie). 
Using Claim~\ref{claim:translation}, 
\[(\proj{x}_\theta \otimes U_\mB) \ket{\phi^+}^{\otimes n}\,=\,\frac{1}{\sqrt{2^n}} \ket{x}_\theta \otimes U_\mB \ket{x}_\theta \,=\, \frac{1}{\sqrt{2^n}}\ket{x}_\theta \otimes \ket{A_{s,s'}}\;,\]
where $A$ is defined from $x,\theta$ and $\mB$ as in~\eqref{eq:def-A}, $s = \sum_{i\in {T}} x_i u_i$ and $s' = \sum_{i\in \ol{T}} x_i u_i$.
Thus 
\begin{align}
 p'_n &= \Es{\theta \in \{0,1\}^n }\Es{\mB} \sum_{x\in \{0,1\}^n} \bra{A_{s,s'}} B^{(\mB,\theta)}_{x_T} \otimes C^{(\mB,\theta)}_{x_{\ol{T}}}\ket{A_{s,s'}}\notag\\
&= \Es{A}\Es{(s,s') A^\perp\times A}  \sum_{u\in A,\, v\in A^\perp} \bra{A_{s,s'}} B^{A}_{u} \otimes C^{A}_{v}\ket{A_{s,s'}}\notag\\
&= \Es{A}\Es{(s,s')\in\{0,1\}^n\times \{0,1\}^n}  \sum_{u\in A,\, v\in A^\perp} \bra{A_{s,s'}} B^{A}_{u} \otimes C^{A}_{v}\ket{A_{s,s'}}\;,\label{eq:qn}
\end{align}
where the second equality is by definition of $B^{\mB,\theta}_{x_T}$ and $C^{(\mB,\theta)}_{x_{\ol{T}}}$ in~\eqref{eq:def-BC} and the expectation is over a uniformly random subspace $A\subseteq \F_2^n$ of dimension $\frac{n}{2}$. Here we used that choosing such an $A$ uniformly at random and returning $(A,A^\perp)$ yields the same distribution as choosing a basis $\mB=\{u_1,\ldots,u_n\}$ and $\theta\in\{0,1\}^n$ such that $|\theta|=\frac{n}{2}$ uniformly at random and returning $(\textrm{Span}\{u_i:\,\theta_i=1\} ,\textrm{Span}\{u_i:\,\theta_i=0\})$.  
In the second line above, the expectation over $s,s'$ is uniform over $s\in A^\perp$ and $s'\in A$, and in the third line it is uniform over $s,s'\in\{0,1\}^n$; equality between the second and third lines follows from the definition of $\ket{A_{s,s'}}$. The expression in~\eqref{eq:qn} is precisely $q'_n$, hence we have shown that $p'_n = q'_n$. Taking the supremum over all strategies in the coset-monogamy game proves the lemma. 
\end{proof}

The following claim is used in the proof of Proposition~\ref{prop:coset-game}. 

\begin{claim}\label{claim:translation}
Let $\{u_1,\ldots,u_n\}$ be a basis of $\F_2^n$ and $U_\mB$ defined as in~\eqref{eq:def-ub}. Let $\{u^1,\ldots,u^n\}$ be its dual basis, \emph{i.e.} $u^i\cdot u_j=\delta_{i,j}$. Let $T \subseteq \{1,\ldots,n\}$ be such that $|T|=n/2$ and $A = \textrm{Span}\{u_i:\, i \in \ol{T}\}$. Let $\theta \in \{0,1\}^n$ be the indicator of $\ol{T}$ and $x\in \{0,1\}^n$. Let $s=\sum_{i\in {T}} x_i u_i$ and $s'=\sum_{i\in \ol{T}} x_i u^i$. Then
\[ \ket{A_{s,s'}} \,=\, U_\mB \ket{x}_\theta\;.\]
\end{claim}

\begin{proof}
First observe that 
\begin{equation}\label{eq:tr-1}
\ket{x}_\theta = X^{x_{{T}}} \ket{x_{\ol{T}}}_\theta = X^{x_{{T}}}Z^{x_{\ol{T}}} \ket{0}_\theta\;.
\end{equation}
Next we verify that for any $x,x'\in\{0,1\}^n$, 
\begin{equation}\label{eq:tr-2}
 U_\mB X^x U_\mB^\dagger \,=\, X^{t} \qquad\text{and}\qquad U_\mB Z^{x'} U_\mB^\dagger \,=\, Z^{t' }\;,
\end{equation}
where $t = \sum_i x_i u_i$ and $t'=\sum_i x'_i u^i$. This completes the proof of the claim as
\begin{align*}
 U_\mB \ket{x}_\theta &= U_\mB X^{x_{{T}}} Z^{x_{\ol{T}}} \ket{0}_\theta\\
&= X^s Z^{s'} U_\mB \ket{0}_\theta\\
&= X^s Z^{s'} \ket{A}\\
&= \ket{A_{s,s'}}\;,
\end{align*} 
where the first line is by~\eqref{eq:tr-1}, the second by~\eqref{eq:tr-2}, the third by definition of $\ket{A}$, $U_\mB$, and $\ket{0}_\theta = \sum_{b\in \{0,1\}^{\ol{T}}} \ket{\sum_i b_i e_i}$, and the last is by definition of $\ket{A_{s,s'}}$. 

It remains to show~\eqref{eq:tr-1}. 
 We show the first relation, the second is analogous. Writing $X^x = \sum_y \ket{x+y}\bra{x}$ and using the definition of $U_\mB$ we get 
\[ U_\mB X^x U_\mB^\dagger \,=\, \sum_y \ket{y'+x}\!\bra{y'}\;,\]
where we defined $y' = \sum_i y_i u_i$ and used linearity. The right-hand side is precisely $X^x$.
\end{proof}

\bibliographystyle{alphaarxiv.bst}
\bibliography{monogamy}

\end{document}